\documentclass[12pt]{amsart}

\usepackage{latexsym,amsfonts,amsmath,epsfig,tabularx,amsthm,dsfont,mathrsfs}

\usepackage{graphicx}
\usepackage{pst-math,pst-xkey}
\usepackage{url}
\usepackage{color}

\theoremstyle{definition}

\newtheorem{thm}{Theorem} 

\newtheorem{lem}{Lemma}

\newtheorem{cor}{Corollary}

\newcommand{\widebar}[1]{\mbox{\kern1.5pt\hbox{\vbox{\hrule height 0.6pt
\kern0.35ex
        \hbox{\kern-0.15em \ensuremath{#1 }\kern0.0em}}}}\kern-0.1pt}

\newcommand{\disp}{{\rm disp}}
\newcommand {\ve} {\varepsilon}

\newcommand{\abs}[1]{\left\vert #1 \right\vert}

\newlength{\fixboxwidth}
\title[On the size of the largest empty box amidst a point set]
{On the size of the largest empty box \\amidst a point set} 

\date{}

\author{Christoph Aistleitner} 
\address{ Institute of Financial Mathematics and applied Number Theory,
University Linz}
\email{aistleitner@math.tugraz.at}

\author{Aicke Hinrichs} 
\address{Institute of Analysis,
University Linz}
\email{aicke.hinrichs@jku.at}

\author{Daniel Rudolf}
\address{Institute for Mathematical Stochastics,
Universit\"at G\"ottingen}
\email{daniel.rudolf@uni-goettingen.de}

\thanks{The first author is supported by Schr\"odinger scholarship J-3311 and by projects I1751-N26 and F5507-N26 of the
Austrian Science Fund (FWF). The third author is supported by the 
DFG Priority Program 1324 and the DFG Research Training Group 2088.}

\begin{document}

\begin{abstract}
The problem of finding the largest empty axis-parallel box amidst a point 
configuration is a classical problem in 
computational geometry.
It is known that the volume of the largest empty box is of asymptotic order $1/n$ for $n \to \infty$ 
and fixed dimension $d$. However, it is natural to assume that the volume of the largest empty box 
increases as $d$ gets larger. In the present paper we prove that this actually is the case: for every 
set of $n$ points in $[0,1]^d$ there exists an empty box of volume at least $c_d n^{-1}$, where $c_d 
\to \infty$ as $d \to \infty$. More precisely, $c_d$ is at least of order roughly $\log d$. 
\end{abstract}

\maketitle

\section{Introduction}

The problem of finding the largest empty axis-parallel rectangle amidst a point
configuration in the unit square is a standard problem in computational geometry
and computational complexity theory. Here the emphasis is on the word
\emph{finding}; that is, researchers are actually interested in the complexity
of algorithms whose output is the largest empty rectangle. The problem has
probably been introduced by Naamad, Lee and Hsu  \cite{nlh}, and generalizes in
a natural way to the multi-dimensional case, where one has to find the largest
empty axis-parallel box amidst a point configuration in the $d$-dimensional unit
cube. Given the prominence of the problem, it is quite surprising that very
little is known on the size (or, more precisely, the area resp. $d$-dimensional
volume) of the largest empty box. There are actually two problems, a ``lower bound problem'' 
and an ``upper bound problem'': one asking for the
minimal size of the largest empty box for \emph{any} point configuration (a kind
of ``irregularities of distributions'' problem), and one asking for the maximal size of the 
largest empty box for an optimal point configuration.

Our attention was drawn to this problem by the fact that questions asking for the
size of the largest empty box have recently appeared in approximation theory in
problems concerning the approximation of high-dimensional rank one tensors; see
\cite{bddg} and \cite{nr}. However, we believe that the problem is very natural
and interesting in its own right. Our main interest was in the multi-dimensional
case, and the main question was: If the dimension $d$ of the problem is
increased, does the minimal size of the largest empty box (taking the minimum
over all configurations of $n$ points) necessarily increase as $d$ increases? As
it turns out, the answer is affirmative. However, we are still far from a
complete quantitative solution of the problem.

To state our results, let us fix the notation. Let $d \geq 2$ and let $[0,1]^d$
be the $d$-dimensional unit cube. For $x,y\in [0,1]^d$
with $x=(x_1,\dots,x_d)$ and $y=(y_1,\dots,y_d)$ we write $x\leq y$ if 
this inequality holds coordinate-wise. For $x\leq y$ we write $[x,y)$ for the
axis-parallel box $\Pi_{i=1}^d [x_i,y_i)$, and
define
\[
 \mathscr{B} = \left\{ [x,y) \colon x,y\in [0,1]^d,\; x\leq y   \right\}.
\]
For $n \geq 1$ let 
$T$ be a set of points in $[0,1]^d$ of cardinality $\abs{T}=n$. The volume of the
largest empty axis-parallel box, which we call the \emph{dispersion} of $T$, is
then given by
\[
 \disp(T) = \sup_{B\in \mathscr{B},\, ~B\cap T= \emptyset} \lambda(B),
\]
where $\lambda(B)$ denotes the volume of $B$. We are in
particular interested in the minimal dispersion of  point sets; thus we set
$$
\disp^*(n,d) = \inf_{T \subset [0,1]^d, ~|T|=n} \quad \disp(T).
$$
It turns out that this minimal dispersion is of asymptotic order $1/n$ as a
function of $n$. Consequently, to capture the dependence of this quantity on the
dimension $d$ we define the number
\begin{equation} \label{cd}
c_d = \liminf_{n \to \infty} n ~\disp^*(n,d).
\end{equation}
Moreover, 
we define the inverse function of the minimal dispersion
\begin{equation} \label{ned}
 N(\varepsilon,d) = \min \left\{ n\colon   \disp^*(n,d) \leq \varepsilon
\right\}, \qquad \textrm{ for $\varepsilon\in (0,1)$.}
\end{equation}
In particular, this number is of 
interest in the applications in approximation theory mentioned
above.

For $\disp^*(n,d)$ one trivially has the lower bound
$$
\disp^*(n,d) \geq \frac{1}{n+1},
$$
which follows directly from splitting the unit cube into $n+1$ parts and using
the pigeon hole principle. On the other hand, the best published upper bound is
\begin{equation}\label{prodo}
\disp^*(n,d) \leq \frac{1}{n} \left( 2^{d-1} \prod_{i=1}^{d-1} p_i \right),
\end{equation}
where $p_i$ denotes the $i$-th prime. This upper bound is essentially due to
Rote and Tichy \cite{rt}, with a detailed proof given in \cite{DuJi13}. Note
that by the prime number theorem the product on the right-hand side of
\eqref{prod} grows super-exponentially in $d$.
After reading a draft version of our manuscript, Gerhard Larcher communicated to us a proof of the upper bound 
\begin{equation} \label{prod}
\disp^*(n,d) \leq \frac{2^{7d+1}}{n}.
\end{equation}
With his permission, we include the argument here in the last section.
This bound is better than \eqref{prodo} for $d\ge 54$.

The only non-trivial lower bound known to date is 
$$
\disp^*(n,d) \geq \frac{5}{4(n+5)},
$$
due to Dumitrescu and Jiang \cite{DuJi13}. Thus for the constant $c_d$ from
\eqref{cd} we have
\begin{equation} \label{cdest}
c_d \in \left[ \frac{5}{4}, ~2^{7d+1} \right], \qquad d
\geq 2.
\end{equation}
It is a natural question to ask whether
$$
c_d \to \infty \qquad \textrm{as $d \to \infty$;}
$$
that is, whether there is some kind of ``irregularity of distributions''-type
behavior in high dimensions.\footnote{As we learned from one referee, 
this problem was also stated as \emph{Open question 6} in 
the \emph{Computational Geometry Column} of Dumitrescu and Jiang; see \cite{cgc}.} 
As we will prove below, the answer is
affirmative.

Before we state our main results, we note that a lower bound of the minimal dispersion
leads also to a lower bound of the minimal extremal discrepancy. Moreover, the point sets leading to the upper bounds \eqref{prodo} and \eqref{prod} are well-known low discrepancy sequences. 
This provides a link of the problem investigated in this paper 
to the theory of uniform distribution modulo one and
discrepancy theory, and to the theory of irregularities of distributions. For
the first two subjects see for example \cite{dt,kn}, for the latter see
\cite{bc}. 
 The theory of digital nets and low-discrepancy sequences, which is
used to prove \eqref{prod}, is explained in detail in \cite{dp}. Finally, the
theory of information-based complexity, where quantities such as the one in
\eqref{ned} are investigated, is described in \cite{nw}.

The topic addressed in this paper has been recently 
taken up by Ullrich \cite{mario}, who considered 
the dispersion on the $d$-dimensional unit \emph{torus} 
rather than on the unit cube, proved lower bounds for this 
notion of dispersion and obtained results which have an interesting 
resemblance of known \emph{inverse of the discrepancy}-type results.

During the refereeing process of the current paper Dumitrescu and Jiang in 
\cite{arxiv_Dumi_Jian} established a generalization of our proof technique. 
It does not
lead to a better bound, but it is
interesting to note that the generalization connects to some classical concepts
in extremal set theory. In particular, the concept of ``properly overlapping
partitions'' turned out to be a rediscovery of the previously known concept of ``qualitative
independent partitions''.

\section{Main result}

In the statement of the following theorem, and in the sequel, $\log_2$ denotes the dyadic logarithm.

\begin{thm}  \label{thm: main}
 For all positive integers $d$ and $n$ we have
 \begin{equation}  \label{eq: low_large_m}
  \disp^*(n,d) \geq \frac{\log_2 d}{4 (n + \log_2 d)}.
 \end{equation}
\end{thm}
We can state this result also in terms of the inverse of the minimal dispersion.
\begin{cor}
 For $\varepsilon \in (0,1/4)$ and $d \geq 1$ we have
 \[
  N(\varepsilon,d) \geq (4\ve)^{-1} (1-4\ve)\log_2 d.
 \]
\end{cor}

For the constant $c_d$ we can directly deduce the following corollary from Theorem
\ref{thm: main}.

\begin{cor} \label{co2}
We have
$$
c_d \geq \frac{\log_2 d}{4} \qquad \textrm{for $d \geq 1$.}
$$
In particular $\lim_{d \to \infty} c_d = \infty$.
\end{cor}

Note that by Corollary \ref{co2} the constant $c_d$ grows at least
logarithmically as $d$ increases; on the other hand, the upper bound for $c_d$
in \eqref{cdest} grows super-exponentially. 
So there remains a large gap between the lower and upper bounds for $c_d$.

\section{Proof and auxiliary lemmas}

The result is trivial for $d=1$. Thus we will always assume in the sequel that $d \geq 2$.

\begin{lem}  \label{lem: basic}
Let positive integers $\ell, n$ be given. Then
    \[
     \disp^*(n,d) \geq \frac{(\ell+1)\disp^*(\ell,d)}{n+\ell+1}.
    \]
\end{lem}
 The proof of the lemma uses an idea similar to the one in the proof of
\cite[Theorem~1]{DuJi13}.
\begin{proof}
 First we note that there is an integer $k \geq 0$
such that
 $n=(\ell+1)k+r$ for $r\in \{0,1,\dots,\ell\}$. We will use the following two
facts:
 \begin{itemize}
  \item For any decomposition of $[0,1]^d$ into $k+1$ boxes of equal volume, there is one box which contains at most $\ell$ points.
  \item For an arbitrary box $A \in \mathscr{B}$, writing 
  \[
   \disp^*(A,\ell,d) := \inf_{T \subset A, ~|T|=\ell} \quad \sup_{B\in
\mathscr{B},~B \subset A, ~B\cap T= \emptyset} \quad \lambda(B)
  \]
  we always have
  \[
   \disp^*(A,\ell,d) \geq \lambda(A) \disp^*(\ell,d).
  \]
  \end{itemize}
By the application of these facts we obtain
\begin{align*}
 \disp^*(n,d) \geq \frac{\disp^*(\ell,d)}{k+1},
\end{align*}
which by $n\geq (\ell+1)k$ proves the assertion.
\end{proof}
The next result tells us that if $d$ is sufficiently large, then for a small
number of
points in $[0,1]^d$ we can find a strong structure in the coordinates of these
points. 

\begin{lem}  \label{lem: crucial}
Let $\ell$ be an positive integer and assume that $d\geq 2^\ell-1$. 
Then 
 \[
  \disp^*(\ell,d) \geq 1/4.
 \]
\end{lem}

\begin{proof}

Let $T=\{t_1,\dots,t_\ell\} \subset [0,1]^d$ 
be an arbitrary point set. To denote the coordinates of a point $t\in [0,1]^d$ we use the notation
$t=(t(1),\dots,t(d))$.
Set 
$$
\chi(x) = \left\{ \begin{array}{ll} 0 & \textrm{if $0 \leq x \leq 1/2$} \\ 1 & \textrm{if $1/2 < x \leq 1$.} \end{array} \right.
$$
Furthermore, for $m \in \{1, \dots, \ell\}$, define
$$
\tau_m = \big( \chi(t_m(1)), \chi(t_m(2)), \dots, \chi(t_m(d)) \big).
$$
If there exists an index $i\in\{1,\dots,d\}$ such that 
\begin{equation} \label{tau0}
(\tau_1(i), \dots, \tau_\ell(i)) = (0, \dots, 0) \qquad \textrm{or} \qquad (\tau_1(i), \dots, \tau_\ell(i)) =  (1, \dots, 1), 
\end{equation}
then all elements of $T$ are contained in one half of $[0,1]^d$, and we have $\disp(T) \geq 1/2$. 
On the other hand, if \eqref{tau0} fails then by $d \geq 2^\ell - 1$ 
and by the pigeon hole principle there must exist two distinct indices $i,j \in \{1, \dots, d\}$ such that 
$$
\big(\tau_1 (i), \dots, \tau_\ell(i)\big) = \big(\tau_1 (j), \dots, \tau_\ell(j)\big).
$$
We assume without loss of generality that $i=1$ and $j=2$. Then we have
$$
T \in \left(\left[0,\frac{1}{2}\right] \times \left[0, \frac{1}{2}\right] \times [0,1)^{d-2}\right) \cup  \left(\left(\frac{1}{2},1 \right] \times \left(\frac{1}{2},1\right] \times [0,1)^{d-2}\right).
$$
Let $0<\ve<\frac12$, and let $A$ denote the box 
$$
A = \left[0,\frac{1}{2}-\ve\right) \times \left[\frac{1}{2}+\ve,1\right] \times [0,1)^{d-2}.
$$
Then
$$
T \cap A = \emptyset,
$$
and 
$$
\lambda(A) = \left(\frac{1}{2} - \ve \right)^2.
$$
Letting $\ve \to 0$ we obtain the conclusion of the lemma.
\end{proof}

 Now we are able to prove Theorem~\ref{thm: main}.
 \begin{proof}[Proof of Theorem~\ref{thm: main}:]
We set 
  \[
   \ell = \left \lfloor \log_2 d \right \rfloor.
  \] 
Then we obviously have $2^\ell \leq d$, which allows us to apply Lemma  \ref{lem: crucial}. Thus by Lemma~\ref{lem: basic} and Lemma \ref{lem: crucial} we have
\begin{eqnarray}
\disp^*(n,d) & \geq & \frac{(\ell+1)\disp^*(\ell,d)}{\ell+1+n} \nonumber\\
   & \geq & \frac{\ell+1}{4 (\ell+1+n)}\\ \nonumber
   & \geq &  \frac{\log_2 d}{4 (n+ \log_2 d)}. \nonumber
\end{eqnarray}
This proves the theorem.
 \end{proof}

\section{Upper bound}

We now present the argument for the upper bound \eqref{prod} discussed in the introduction. 
This argument was shown to us by Gerhard Larcher, 
who generously allowed us to include it here.

A dyadic interval is an interval of the form $[a 2^{-k},(a+1)2^{-k})$ 
for non-negative integers $a,k$ with $0\le a < 2^k$. 
A ($d$-dimensional) dyadic box is the Cartesian product of $d$ dyadic intervals.
Let $t$ be a non-negative integer and let $m$ be a positive integer 
with $t\le m$. A $(t,m,d)$-net (in base 2) is a set $T$ of $2^m$ points 
in $[0,1]^d$ such that each dyadic box of volume $2^{t-m}$ contains exactly 
$2^t$
points of $T$. 
Since any interval $[x,y)\subset [0,1)$ contains a dyadic 
interval of length at least $4^{-1}(y-x)$, any box $B\in \mathscr{B}$ 
contains a dyadic box of volume at least $4^{-d} \lambda(B)$. It follows that, 
if $T$ is a $(t,m,d)$-net and $4^{-d} \lambda(B) \ge 2^{t-m}$, then $B$ 
contains a point of $T$. Hence
$$ \disp^*(2^m,d) \le \disp(T) \le 2^{t-m+2d}.$$
In \cite{nx}, $(t,m,d)$-nets with $t\le 5d$ are constructed for all 
dimensions $d$ and $m\ge t$. If now $m$ is chosen such that $2^m \le n < 2^{m+1}$, 
we arrive at
$$ \disp^*(n,d) \le \disp^*(2^m,d) 
\le 2^{7d-m}\le 
2^{7d+1}/n.
$$
This shows \eqref{prod} in the case $n\ge 2^{5d}$. For $n<2^{5d}$ 
the same bound holds trivially.


\end{document}